\documentclass[envcountsame]{llncs}
\usepackage{graphicx}
\usepackage{float}

\usepackage{enumerate}
\usepackage[ruled,vlined]{algorithm2e}
\usepackage{graphicx,color}
\usepackage[dvipsnames]{xcolor}
\usepackage{fancybox}
\usepackage[caption=false]{subfig}
\usepackage{tikz}
\usetikzlibrary{positioning,chains,fit,shapes,calc}
\usetikzlibrary{trees}
\usetikzlibrary{decorations.pathmorphing}
\usetikzlibrary{decorations.markings}
\usepackage{cool}
\Style{DSymb={\mathrm d},DShorten=true,IntegrateDifferentialDSymb=\mathrm{d}}

\newcommand{\abs}[1]{\left\vert#1\right\vert}
\newcommand{\set}[1]{\left\{#1\right\}}
\newcommand{\eps}{\varepsilon}

\newlength{\probwidth}
\newtheorem{fact*}{Fact}
\usepackage{todonotes}

\newcommand{\BIS}{\#\mathtt{BIS}}
\newcommand{\tree}[1]{\mathbb{T}_{#1}}
\newcommand{\bianti}{\mbox{\scshape Bi-Anti-Ferro}}
\newcommand{\biferro}{\mbox{\scshape Bi-Ferro}}
\newcommand{\spin}{\mbox{\scshape Spin}}
\newcommand{\anti}{\mbox{\scshape Anti-Ferro}}
\newcommand{\ferro}{\mbox{\scshape Ferro}}

\begin{document}

\title{The Complexity of Ferromagnetic Two-spin Systems with External Fields}

\author{ Jingcheng Liu\inst{1} \and Pinyan Lu\inst{2} \and Chihao Zhang\inst{1}}

\institute{ Shanghai Jiao Tong University, \email{\{liuexp, chihao.zhang\}@gmail.com} \and Microsoft Research. \email{pinyanl@microsoft.com} }

\date{}

\maketitle

\begin{abstract}
  We study the approximability of computing the partition function for ferromagnetic two-state spin systems.
  The remarkable algorithm by Jerrum and Sinclair showed that there is a fully polynomial-time randomized approximation scheme (FPRAS) for the special ferromagnetic Ising model with any given uniform external field.  Later, Goldberg and Jerrum proved that it is $\BIS$-hard for Ising model if we allow inconsistent external fields on different nodes.
  In contrast to these two results, we prove that for any ferromagnetic two-state spin systems except the Ising model,
  there exists a threshold for external fields beyond which the problem is $\BIS$-hard, even if the external field is uniform.

\end{abstract}

\section{Introduction}
Spin systems are well studied in statistical physics and applied probability. We focus on two-state spin systems in this paper.
An instance of a spin system is described by a graph $G(V,E)$, where vertices are particles and edges indicate neighborhood relation among them.  A configuration $\sigma: V\rightarrow \{0, 1\}$ assigns one of the two states to every vertex.
The contribution of local interactions between adjacent vertices is quantified by a
matrix $\mathbf{A}=\begin{bmatrix} A_{0,0} & A_{0,1} \\ A_{1,0} & A_{1,1} \end{bmatrix}=\begin{bmatrix} \beta & 1 \\ 1 & \gamma \end{bmatrix}$, where $\beta, \gamma \geq 0$.
The contribution of vertices in different spin states is quantified by a vector $\mathbf{b}=\begin{bmatrix} b_{0}  \\ b_1 \end{bmatrix}=\begin{bmatrix} \mu \\ 1 \end{bmatrix}$, where  $\mu > 0$.
This $\mu$ is also called the external field of the system, which indicates a priori preference of an isolate vertex.
The \emph{partition function} $Z_{(\beta, \gamma, \mu)}(G)$ of a spin system $G(V,E)$ is defined to be the following exponential summation:
\[ Z_{(\beta, \gamma, \mu)}(G)=\sum_{\sigma \in \{0,1\}^{V}} \prod_{v \in V} b_{\sigma_v} \prod_{(u,v) \in E} A_{\sigma_u, \sigma_v}.\]
We call such a spin system parameterized by $(\beta, \gamma, \mu)$. If the parameters are clear from the context, we shall write $Z(G)$ for short.
Although originated from statistical physics, the spin model is also accepted in computer science as a framework for counting problems.
For example, with $\beta=0$, $\gamma=1$ and $\mu=1$, $Z_{(\beta, \gamma, \mu)}(G)$ is the number of independent sets (or vertex covers) of the graph $G$.

Given a set of parameters  $(\beta, \gamma, \mu)$, it is a computational problem to compute the partition function  $Z_{(\beta, \gamma, \mu)}(G)$ where graph $G$ is given as input.
We denote this computation problem as $\spin(\beta, \gamma, \mu)$ and
want to characterize their computational complexity in terms of $\beta$, $\gamma$ and $\mu$.
For \emph{exact} computation, polynomial time algorithms are known only for the very restricted settings that $\beta \gamma =1$ or $(\beta,\gamma)= (0,0)$, and for all other settings the problem is proved to be \#P-hard \cite{BulatovG05}. Therefore, the main focus is to study their approximability.
For any given parameter $\eps>0$, the algorithm outputs a number $\hat{Z}$ such that $ Z(G) \exp (-\eps ) \leq \hat{Z} \leq Z(G) \exp (\eps )$
 and runs in time $poly(n,1/\eps)$, where $n$ is the size of the graph $G$.
This is called a fully polynomial-time approximation scheme (FPTAS). The randomized relaxation of FPTAS is called fully polynomial-time randomized approximation scheme (FPRAS), which uses random bits and only requires the final output be within the required accuracy with high probability.

The spin systems $(\beta, \gamma, \mu)$ are classified into two families with distinct physical and computational properties: \emph{ferromagnetic} systems ($\beta \gamma >1$) and \emph{anti-ferromagnetic} systems ($\beta \gamma < 1$).
We shall denote the corresponding computation problems respectively by $\ferro(\beta, \gamma, \mu)$ and $\anti(\beta, \gamma, \mu)$, so as to emphasize which family these parameters belong to.
Systems with $\beta \gamma =1$ are degenerate and trivial both physically and computationally. As a result, we only study systems with $\beta \gamma \neq 1$.

Great progress has been made recently for approximately computing the partition function for anti-ferromagnetic two-spin systems: it admits an FPTAS up to the uniqueness threshold~\cite{Weitz06,LLY12,SST,LLY13}, and is NP-hard to approximate in the non-uniqueness range~\cite{SS12,galanis2012inapproximability}.
The uniqueness threshold is a phase transition boundary in physics. It is widely conjectured that the computational difficulty is related to the phase transition point in many problems; this is one of the very few examples where a rigorous proof was obtained.

For ferromagnetic systems, the picture is quite different. The uniqueness condition does not coincide with the transition of computational difficulty and it is not clear whether it plays any role in the computational difficulty.
In a seminal paper~\cite{JS93}, Jerrum and Sinclair gave an FPRAS for ferromagnetic Ising model $\beta=\gamma >1$ with any external field $\mu$.
Thus, there is no transition of computational difficulty for ferromagnetic Ising model, which contrasts the situation for anti-ferromagnetic Ising model $\beta=\gamma <1$.
For general ferromagnetic  spin systems with external field, the approximability is less clear. Since the Ising model ($\beta=\gamma$) is solved, we focus on the case $\beta \neq \gamma$ and always assume $\beta<\gamma$ by symmetry.
 By transferring to Ising model, an FPRAS was known for the range of  $\mu\le\sqrt{\gamma/\beta}$~\cite{goldberg2003computational}.

 On the other hand, a hardness result was obtained for Ising model with inconsistent external fields~\cite{goldberg2007complexity}.
 This is a generalization of the spin system where the external fields for different vertices can be different and taken from a set $\mathcal{V}$.
 We use $\spin (\beta, \gamma, \mathcal{V})$ ($\ferro (\beta, \gamma, \mathcal{V})$ or $\anti (\beta, \gamma, \mathcal{V})$ ) to denote this computation problem.
   It is proved that the Ising model with arbitrary external fields $\ferro (\beta, \beta, (0, +\infty) )$  is $\BIS$-hard, namely the problem is at least as hard as counting independent sets on bipartite graphs ($\BIS$).
   $\BIS$ is a problem of intermediate hardness and has been conjectured to admit no FPRAS \cite{dicho_DGJ10}.
   The reduction used here is called approximation-preserving reduction as introduced in \cite{dyer2004relative}: Let $A,B:\Sigma^*\to \mathbb{R}$ be two functions.
An \emph{approximation-preserving reduction} from $A$ to $B$ is a randomized algorithm that approximates $A$ while using an oracle for $B$.
We write $A\le_{AP} B$ for short if an approximation-preserving reduction exists from $A$ to $B$.
To get that $\BIS$-hardness result, one need to use (or simulate) both arbitrarily small and large external fields.
As $\beta<\gamma$, we can always simulate some arbitrarily small external fields with gadgets.
However, simulating arbitrarily large external fields is only possible when $\beta \mu +1 > \mu+\gamma$, in which case one gets a $\BIS$-hardness result similarly.
If this is not the case, and in particular if $\beta\leq 1 <\gamma$, no hardness result was known for any bounded external fields.
These systems have certain monotonicity property, so all external fields that can be simulated by gadgets are inherently bounded by above.  It was not even clear whether there is any hardness result or not.
As a first result, we show that the problem is already hard as long as we allow sufficiently large (yet still bounded by above) and vertex-dependent external fields.

\begin{theorem}\label{thm:1}
	For any $\beta<\gamma$ with $\beta \gamma >1$, there exsits a bounded set $\mathcal{V}$ such that $\ferro (\beta, \gamma, \mathcal{V})$ is $\BIS$-hard.
\end{theorem}

The main difficulty is for the case of $\beta\leq 1$, for which we cannot simulate any external field larger than the upper bound of $\mathcal{V}$.
We overcome this difficulty by making use of a recent beautiful result in~\cite{cai2013approximating}.
Instead of starting with the independent set problem on arbitrary bipartite graphs, we start with a soft ($\beta \gamma >0$) anti-ferromagnetic two-spin system on bounded degree bipartite graphs from~\cite{cai2013approximating}.
Then all the external fields needed for the reduction are bounded.

However, in the above reduction, we do need vertices to have different external fields to make the reduction go through.
 This gives a hardness result for  $\ferro (\beta, \gamma, \mathcal{V})$ but not $\ferro (\beta, \gamma, \mu)$ for a single $\mu$.
 It is more interesting and intriguing (both physically and computationally) to understand the computational complexity of a uniform spin system  $(\beta, \gamma, \mu)$ with the same external field $\mu$ on all the vertices. As our main result of this paper, we also prove $\BIS$-hardness on this uniform case for sufficiently large single external field $\mu$.
We prove that when $\mu$ is sufficiently large, we can realize by sufficient precision of all the external fields which is smaller than $\mu^*(\mu, \beta, \gamma)$, where  $\mu^*(\mu, \beta, \gamma)$ is a function of   $\mu, \beta$ and  $\gamma$,  and approaches infinity as $\mu$ goes to infinity. Then by choosing large enough $\mu$ and making use of Theorem \ref{thm:1}, we get our main theorem.
\begin{theorem}
	\label{thm:main-thm}
For any $\beta<\gamma$ with $\beta \gamma >1$, there exist a $\mu_0$ such that  $\ferro(\beta, \gamma, \mu)$ is $\BIS$-hard  for all $\mu\geq \mu_0$.
\end{theorem}

Our main technical contribution is the construction of a family of gadgets to simulate a given target external field. We use a reverse idea of correlation decay to do that.
Correlation decay is proved to be a very powerful technique to design FPTAS for counting problems (see for examples~\cite{Weitz06,BG08,LLY13,SST,counting-edge-cover,fibo-approx}). In correlation decay based FPTAS, one first establishes a tree structure and
hope to compute the marginal probability of the root.
With a recursive relation, one can write the marginal probability of the root as a function of that of its sub-trees, then truncate the computation tree at certain depth and do a rough guess at the leaf nodes.
A correlation decay property ensures that the error for the root is exponentially small although there are
constant error for the leaves.
Here, we use a similar idea to construct a tree gadget so that the marginal probability for the root is very close to a target value.
Using the same recursion, one translates the target marginal probability for the root to that of its sub-trees.
In the leaf nodes, we simply use some basic gadgets to approximate the target marginal probabilities.
Again, although these approximation for leaves may have constant gap, the error at the root is exponentially small thanks to the correlation decay property.
We believe that this idea of using an algorithm design technique to build gadgets and get hardness result is  of independent interest and may find applications in other problems.

We also make some improvements on the algorithm side showing that there is an FPRAS if $\mu\le \gamma/\beta$.
We remark that all the computational problem  $\ferro(\beta, \gamma, \mu)$  and $\ferro (\beta, \gamma, \mathcal{V})$
is no more difficult than $\BIS$, as we can use the standard transformation to transform any ferromagnetic two-spin system to
ferromagnetic Ising model with possibly different external fields and use the  $\BIS$-easiness result in~\cite{goldberg2007complexity}. Thus, the two $\BIS$-hardness theorems can also be stated as $\BIS$-equivalent.
We believe that the conjecture here is that for any fixed $\beta<\gamma$, there exists a critical $\mu_c$ such that it admits an FPRAS if the external field $\mu<\mu_c$, and it is $\BIS$-equivalent if $\mu>\mu_c$.
The result of this paper is an important step towards this dichotomy.

\subsection{Related works}
The approximation for partition function has been studied extensively with both positive and negative results~\cite{Weitz06,BG08,LLY13,SST,JS93,app_JSV04,col_Vigoda99,goldberg2012approximating,SS12,galanis2012inapproximability}.
 For the algorithm side of ferromagnetic two-spin systems, besides the FPRASes, there is also a recent deterministic FPTAS for certain range of the parameters based on correlation decay and holographic reduction~\cite{fibo-approx}.
%



\section{Bounded Local Fields}
In the section, we show that bounded local fields are sufficient to establish a hardness result. The following theorem is a formal statement of Theorem \ref{thm:1}.
 \begin{theorem}
	\label{thm:existence}
  Let $\beta<\gamma$, $\beta \gamma >1$,  $\Delta= \lfloor \frac{\sqrt{\beta\gamma}+1}{\sqrt{\beta\gamma}-1} \rfloor + 1$  and $\mu>\left(\sqrt{\frac{\gamma}{\beta}}\right)^\Delta$.
   Then \\
   $\ferro\left(\beta, \gamma, [1,\mu]\right)$ is $\BIS$-hard.\footnote{Technically, we should only define the problem by a finite set of external fields.
	   In this paper and as in many others, we adopt the following convention: when we say a problem with an infinite set of external fields is hard, it means that there exists a finite subset of external fields to make the problem hard already.
   }
\end{theorem}
  We first introduce our starting point from anti-ferromagnetic Ising model on bipartite graphs, and show the reduction in the second subsection.

\subsection{Anti-ferromagnetic Spin Systems on Bipartite Graphs }
$\BIS$ is a special anti-ferromagnetic two-state spin system. Similar to $\BIS$, one can also study other
anti-ferromagnetic two-state spin systems on bipartite graphs.
We use a prefix {\scshape Bi-} to emphasize that input graphs are bipartite, and a subscript $\Delta$ to indicate that maximum degree is $\Delta$.
For instance, the problem of $\anti(\beta,\gamma,\mu)$ on bipartite graphs with maximum degree $\Delta$, is denoted shortly by $\bianti_{\Delta}(\beta,\gamma,\mu)$.
 The following theorem from~\cite{cai2013approximating} is the starting point of our reduction.

\begin{theorem}[\cite{cai2013approximating}]\label{thm:anti}
  Suppose a set of anti-ferromagnetic parameters $(\beta,\gamma,\mu)$ lies in the non-uniqueness region of the infinite $\Delta$-regular tree $\tree{\Delta}$ and that $\sqrt{\beta \gamma}\ge\frac{\sqrt{\Delta-1}-1}{\sqrt{\Delta-1}+1}$, and $\beta \ne \gamma$ or $\mu\ne 1$.
  Then $\bianti_{\Delta}(\beta,\gamma,\mu)$ is $\BIS$-hard.
\end{theorem}

For simplicity, we use the special anti-ferromagnetic Ising model $\beta=\gamma<1$ in our reduction,
where the non-uniqueness condition is easy to state.

\begin{proposition}
  \label{thm:critical}
  If $\beta<\frac{\Delta-1}{\Delta+1}$, then there exists a critical activity $\mu_c(\beta,\Delta)> 1$ such that the Gibbes measure of Ising model $(\beta,\beta,\mu)$ on infinite $\Delta$-regular tree $\tree{\Delta}$ is unique if and only if $|\log \mu|\ge \log \mu_c(\beta,\Delta)$.
\end{proposition}

Proposition \ref{thm:critical} is folklore, a proof can be found in, e.g. \cite{SST}. Combining these two results, we can get

\begin{corollary}\label{cor:hardness}
For all $0<\beta<1$,
there is an $\eps>0$  such that for any $\mu \in (1, 1+\eps)$, $\bianti_{\Delta}(\beta,\beta,\mu)$ is $\BIS$-hard, where $\Delta=\lfloor\frac{1+\beta}{1-\beta}\rfloor+1$.
\end{corollary}
\begin{proof}
As $\Delta=\lfloor\frac{1+\beta}{1-\beta}\rfloor+1$, we know that $\beta<\frac{\Delta-1}{\Delta+1}$. Then by Proposition \ref{thm:critical}, we can choose $\eps=\mu_c(\beta,\Delta)-1$ and get that
 $(\beta,\beta,\mu)$  is in the non-uniqueness region of the infinite $\Delta$-regular tree $\tree{\Delta}$ for all $\mu \in (1, 1+\eps)$. In order to make use of Theorem \ref{thm:anti} and conclude our proof, we only need to verify that $\beta \geq \frac{\sqrt{\Delta-1}-1}{\sqrt{\Delta-1}+1}$. Our choice of $\Delta$ is the smallest integer to satisfy $\beta<\frac{\Delta-1}{\Delta+1}$. As a result, we have
$\beta\geq \frac{\Delta-1-1}{\Delta-1+1} \geq \frac{\sqrt{\Delta-1}-1}{\sqrt{\Delta-1}+1}$.
\qed
\end{proof}

\subsection{The Reduction}

\begin{lemma}\label{lem:reduction}
For any $\beta<\gamma$ with $\beta \gamma>1$, $\mu>1$ and integer $\Delta>1$, we have
\[ \bianti_{\Delta}\left(\frac{1}{\sqrt{\beta \gamma}},\frac{1}{\sqrt{\beta \gamma}}, \mu\right) \le_{AP} \biferro_{\Delta}\left(\beta,\gamma,\left[\frac{1}{\mu}\sqrt{\frac{\gamma}{\beta}}, \mu \left(\sqrt{\frac{\gamma}{\beta}}\right)^\Delta\right]\right) .\]
\end{lemma}

\begin{proof}

Let bipartite graph $G(L\cup R,E)$ be an instance of anti-ferromagnetic Ising $\left(\frac{1}{\sqrt{\beta \gamma}},\frac{1}{\sqrt{\beta \gamma}}, \mu\right)$ with maximum degree $\Delta$.
We construct an instance of  ferromagnetic system with exactly the same graph.
Each vertex $u \in L$ with degree $d_u$ has weight $\mu\left(\sqrt{\frac{\gamma}{\beta}}\right)^{d_u}$,
and each vertex $v\in R$ has weight $\frac{1}{\mu} \left(\sqrt{\frac{\gamma}{\beta}}\right)^{d_v}$.
Then the maximum possible external field is $\mu \left(\sqrt{\frac{\gamma}{\beta}}\right)^\Delta$ while the minimum one is $\frac{1}{\mu}\sqrt{\frac{\gamma}{\beta}}$. Therefore, it is indeed an instance of
$\biferro_{\Delta}\left(\beta,\gamma,\left[\frac{1}{\mu}\sqrt{\frac{\gamma}{\beta}}, \mu \left(\sqrt{\frac{\gamma}{\beta}}\right)^\Delta\right]\right)$.

Let $Z_1(G)$ be the partition function of the anti-ferromagnetic Ising instance, and $Z_2(G)$ be that for the ferromagnetic system.
We shall prove  that $Z_1(G) = \gamma^{-|F|}\mu^{|R|} Z_2(G)$. Let $V \triangleq L \cup R$, $
A=\left[
  \begin{array}{cc}
    \frac{1}{\sqrt{\beta \gamma}} & 1\\
    1 & \frac{1}{\sqrt{\beta \gamma}}
  \end{array}
\right]
$,
$
A'=\left[
  \begin{array}{cc}
    \sqrt{\frac{\gamma}{\beta}} & \gamma\\
    \gamma &\sqrt{\frac{\gamma}{\beta}}
  \end{array}
\right]
$,
$
\hat A'=\left[
  \begin{array}{cc}
    1 & \beta\\
    \gamma & 1
  \end{array}
\right]
$ and
$
\hat A=\left[
  \begin{array}{cc}
    \beta & 1\\
    1 & \gamma\\
  \end{array}
\right]
$. Then
\begin{align*}
 Z_2(G)
  &=\sum_{\sigma\in \{0,1\}^V}\prod_{(u,v)\in E}\hat A_{\sigma_u,\sigma_v}\prod_{u\in L}\left(\mu\left(\sqrt{\frac{\gamma}{\beta}}\right)^{d_u}\right)^{1-\sigma_u}\prod_{v\in R} \left( \frac{1}{\mu} \left(\sqrt{\frac{\gamma}{\beta}}\right)^{d_v}\right)^{1-\sigma_v}\\
  &=\sum_{\sigma\in \{0,1\}^{V}}\prod_{(u,v)\in E}\hat A'_{\sigma_u,\sigma_v}\prod_{u\in L}\left(\mu\left(\sqrt{\frac{\gamma}{\beta}}\right)^{d_u}\right)^{1-\sigma_u}\prod_{v\in R} \left( \frac{1}{\mu} \left(\sqrt{\frac{\gamma}{\beta}}\right)^{d_v}\right)^{\sigma_v}\\
  &=\sum_{\sigma\in \{0,1\}^{V}}\prod_{(u,v)\in E}A'_{\sigma_u,\sigma_v}\prod_{u\in L}\mu^{1-\sigma_u}\prod_{v\in R}  \frac{1}{\mu^{\sigma_v}}\\
  &=\mu^{-\abs{R}} \gamma^{\abs{F}} \sum_{\sigma\in \{0,1\}^{V}}\prod_{(u,v)\in E}A_{\sigma_u,\sigma_v}\prod_{u\in L}\mu^{1-\sigma_u}\prod_{v\in R} \mu^{1-\sigma_v}\\
  &=\mu^{-\abs{R}} \gamma^{\abs{F}} Z_{1}(G).
\end{align*}
Thus we can get an approximation for the anti-ferromagnetic Ising model by an oracle call to the ferromagnetic two-spin system. This concludes the proof.
\qed
\end{proof}

Now, given the target $\mu>\left(\sqrt{\frac{\gamma}{\beta}}\right)^\Delta $ in Theorem \ref{thm:existence},
we simply choose a $\mu'$ close enough to $1$ in Lemma \ref{lem:reduction} and Corollary \ref{cor:hardness}, such that
$\left[\frac{1}{\mu'}\sqrt{\frac{\gamma}{\beta}}, \mu' \left(\sqrt{\frac{\gamma}{\beta}}\right)^\Delta\right] \subseteq [1, \mu]$ and
$\BIS \leq_{AP} \bianti_{\Delta}\left(\frac{1}{\sqrt{\beta \gamma}},\frac{1}{\sqrt{\beta \gamma}}, \mu'\right) $. Then we can conclude that
$\BIS \leq_{AP} \biferro_{\Delta}(\beta,\gamma, [1, \mu])$ and finish the proof of Theorem \ref{thm:existence}.

\section{Uniform Local Field}
We establish Theorem \ref{thm:main-thm} in this section.
If $\beta>1$, then one can use external field of $\mu>\frac{\gamma-1}{\beta-1}$ to simulate any external fields and get the $\BIS$-hardness result.
This follows from a similar argument as that in~\cite{goldberg2007complexity}. To be self-contained, we also include a formal proof in the appendix. So we assume $\beta\le 1$ in this section.
We also introduce a function $h(x)=\frac{\beta x+1}{x+\gamma}$ which is used throughout this section.
Note that since $\beta \gamma > 1$, $h(x)$ is monotonically increasing and $\frac{1}{\gamma} < h(x) < \beta\le 1$ for $x \in (0,+\infty)$.
We shall prove the following key reduction.

\begin{lemma}\label{lemma:key-reduction}
  Let $\beta\leq 1, \beta \gamma >1$, $d$ be an integer such that $\beta (\beta \gamma)^d > 1$,  $\mu^*$ be the largest solution of $x$ to $x=\mu h(x)^d$, and $\mu > \frac{\gamma^d(\beta \gamma -1 )}{\beta} \left(1+\frac{d+1}{ \ln \left(\beta (\beta \gamma)^d\right)}\right)$. Then $\ferro\left(\beta,\gamma,\left[1, \mu^*\right]\right)$ $\le_{AP}\ferro(\beta,\gamma,\mu)$.
\end{lemma}

As $\mu^*=\mu h(\mu^*)^d$ and $\frac{1}{\gamma} < h(\mu^*) < \beta$,
we have the following bound for $\mu^*$.
\begin{proposition}
  $\frac{\mu}{\gamma^d} < \mu^*< \beta^d \mu$.
\end{proposition}

With this bound and Lemma \ref{lemma:key-reduction}, we can choose sufficiently large $\mu$ so that this $\mu^*$ is large enough to apply the hardness result (Theorem \ref{thm:existence}) of $\ferro(\beta,\gamma,[1, \mu^*])$ to get the hardness result for $\ferro(\beta,\gamma,\mu)$. Formally, we have

\begin{theorem}
  Let $\beta\leq 1, \beta \gamma >1$, $d$ be an integer such that $\beta (\beta \gamma)^d > 1$, $\Delta= \lfloor \frac{\sqrt{\beta\gamma}+1}{\sqrt{\beta\gamma}-1} \rfloor + 1$, and
  $\mu > \gamma^d \max\set{ \left(\sqrt{\frac{\gamma}{\beta}}\right)^\Delta,  \frac{\beta \gamma -1}{\beta} \left(1+\frac{d+1}{ \ln \left(\beta (\beta \gamma)^d\right)}\right)}$.  Then $\ferro(\beta,\gamma,\mu)$ is $\BIS$-hard.
\end{theorem}

We remark that there always exists such integer $d$ since $\beta>0$ and  $\beta \gamma >1$.
Different $d$s give different bounds for $\mu$ and it is not necessarily monotone.
For a given $\beta,\gamma$, one can choose a suitable $d$ to get the best bound\footnote{We give one numerical example here to get some idea of this bound: if $\beta=1$ and $\gamma=2$, we can get $\Delta=6$ and choose $d=1$; then the theorem tell us that the problem $\ferro(1,2,\mu)$ is $\BIS$-hard if $\mu>12$.}.

In the remaining of this section, we prove the key reduction stated in Lemma \ref{lemma:key-reduction}. The main idea is to simulate any external field in $[1,\mu^*]$ by a vertex weight gadget.
In the first subsection, we state the general framework of such simulation. 
Then in the second subsection, we present the detailed construction of a gadget.

\subsection{Vertex Weight Gadget}
\begin{definition}[Vertex weight gadget]
  Let $G(V,E)$ be a graph with a special output vertex $v^*$, define $\mu(G)=\frac{Z_G(v^*=0)}{Z_G(v^*=1)}$ where $Z_G(v^*=0)$ (resp. $Z_G(v^*=1)$) is the partition function of $G(V,E)$ in $(\beta,\gamma,\mu)$-system conditioned on $v^*=0$ (resp. $v^*=1$). We call $G$ a {\rm vertex weight gadget} that realizes $\mu(G)$.
\end{definition}

We also use a family of graphs to approach a given external field. Let $\{G_i\}_{i\ge 1}$ be a family of vertex weight gadgets. We say $\{G_i\}$ realizes $\mu$ if $\lim_{i\to\infty}\mu(G_i)=\mu$.

Vertex weight gadgets can be used to  simulate external fields. Formally, we have the following reductions.

\begin{lemma}
  Let $G$ be a vertex weight gadget  of $(\beta, \gamma, \mathcal{V})$. Then $\spin(\beta, \gamma, \mathcal{V} \cup \{\mu(G) \}) \le_{AP} \spin(\beta, \gamma, \mathcal{V}) $.

  Let $\{G_i\}$ be a sequence of vertex weight gadget of $(\beta, \gamma, \mathcal{V})$ to realize $\mu$ such that for any $\eps >0$ there is a $G_i$ of size $poly\left(\eps^{-1}\right)$ with $\exp (-\eps) \leq \frac{\mu(G_i)}{\mu}\leq \exp (\eps)$. Then  $\spin\left(\beta, \gamma, \mathcal{V} \cup \{\mu \}\right) \le_{AP}\spin(\beta, \gamma, \mathcal{V})  $.
\end{lemma}

\begin{proof}
  The proof of the first part is straightforward. For any instance $H$ of $\spin(\beta, \gamma, \mathcal{V} \cup \{\mu(G) \})$ and a vertex of $H$ with external field $\mu(G)$, we use one copy of $G$ and identify the output vertex of $G$ with that chosen vertex of $H$.
  After the identification, the external field in that vertex is that of output vertex of $G$.
  Therefore, after the modification, the new instance is an instance of $\spin(\beta, \gamma, \mathcal{V})$ and the partition function is equal to the partition function of $H$ scaled by a polynomial-time computable global factor $\left(\frac{Z(G)}{1+\mu(G)}\right)^j$, where $j$ is the number of vertices with external field $\mu(G)$ in $H$.

  For the second part, for an instance $H$ of $\spin(\beta, \gamma, \mathcal{V} \cup \{\mu\})$ and required approximation parameter $\eps$, choose a gadget $G_i$ which is $\eps'=\frac{\eps}{2 n}$ close to realize $\mu$; do the same modification as above using this $G_i$ and call the oracle for the new instance with approximation parameter $\eps'$. This gives the desired approximation for the original instance.
  \qed
\end{proof}

\subsection{The Construction}
We first define a gadget operation $\mathtt{comb}$ as follows:  for a given list of graphs $\mathcal{G}=\{G_1,\dots,G_k\}$, each with output $v_i^*$ for $i\in[k]$, $\mathtt{comb}(\mathcal{G})$ is a new graph $G(V,E)$ that combines the graphs and joins their outputs.
Fig. \ref{fig:tree} is an illustration of $\mathtt{comb}$.
Formally, we define $V=\{u\}\cup\bigcup_{i\in[k]}V(G_i)$ and $E=\{(u,v_i^*)\mid i\in[k]\}\cup\bigcup_{i\in[k]}E(G_i)$, where $u$ is the output of $G$. It is easy to verify that
$\mu(G)=\mu\prod_{i\in[k]}h\left(\mu(G_i)\right)$.

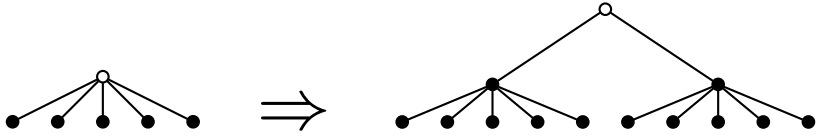
\begin{figure}[htp]
	\centering
	\subfloat{
		\centering
		\begin{tikzpicture}[
				thick,
				level/.style={level distance=6mm, sibling distance=0.6cm}
			]
			\node [circle,draw,scale=0.5]{}
					child {
						node[circle,draw,scale=0.5,fill] {}
						edge from parent
					}
					child {
						node[circle,draw,scale=0.5,fill]{}
						edge from parent
					}
					 child {
						node[circle,draw,scale=0.5,fill]{}
						edge from parent
					}
					 child {
						node[circle,draw,scale=0.5,fill]{}
						edge from parent
					}
					 child {
						node[circle,draw,scale=0.5,fill]{}
						edge from parent
					};
		\end{tikzpicture}
	}
	\qquad
	\scalebox{3}{$\Rightarrow$}
	\qquad
	\subfloat{
		\centering
		\begin{tikzpicture}[
				thick,
				level/.style={level distance=10mm/#1, sibling distance=3cm/(5^(#1-1))}
			]
			\node [circle,draw,scale=0.5]{}
				child {
				 node[circle,draw,scale=0.5,fill]{}
					child {
						node[circle,draw,scale=0.5,fill] {}
						edge from parent
					}
					child {
						node[circle,draw,scale=0.5,fill]{}
						edge from parent
					}
					 child {
						node[circle,draw,scale=0.5,fill]{}
						edge from parent
					}
					 child {
						node[circle,draw,scale=0.5,fill]{}
						edge from parent
					}
					 child {
						node[circle,draw,scale=0.5,fill]{}
						edge from parent
					}
					edge from parent
				}
				child {
				 node[circle,draw,scale=0.5,fill]{}
					child {
						node[circle,draw,scale=0.5,fill] {}
						edge from parent
					}
					child {
						node[circle,draw,scale=0.5,fill]{}
						edge from parent
					}
					 child {
						node[circle,draw,scale=0.5,fill]{}
						edge from parent
					}
					 child {
						node[circle,draw,scale=0.5,fill]{}
						edge from parent
					}
					 child {
						node[circle,draw,scale=0.5,fill]{}
						edge from parent
					}
					edge from parent
				};

		\end{tikzpicture}
	}
	\caption{Result of $\mathcal{S}_5 \Rightarrow \mathtt{comb}\left(\set{\mathcal{S}_5, \mathcal{S}_5}\right)$, the output vertex is marked as unfilled.}
\label{fig:tree}
\end{figure}

We also define two basic gadgets.
Let $\mathcal{S}_w$ be a $w$-star graph, with output being its center. In particular, $\mathcal{S}_0$ is the singleton graph.
Note that $\mu(\mathcal{S}_w) = \mu h(\mu)^w$.
We also define $\mathcal{T}_t$ be a $d$-ary tree with depth $t$.
For any external field $\hat{\mu} \in (0, \mu^*]$, we shall construct a list of gadgets to simulate it.
The two boundaries are approached by $\mathcal{S}_w$ and $\mathcal{T}_t$ respectively.

\begin{proposition}
  \label{prop:convergence}
  Let $\mathcal{T}_t$ be a $d$-ary tree with depth $t$ and $\mathcal{S}_w$ be a $w$-star. Then
  \begin{enumerate}[(1)]
  \item $\set{\mathcal{S}_w}_{w \ge 1}$ realizes $0$, or formally,
    $ \mu(\mathcal{S}_w) = \mu h(\mu)^w < \mu \beta^w.$
  \item
    $\set{\mathcal{T}_t}_{t \ge 0}$ realizes $\mu^*$, or formally, there exist two positive constants $\iota$ and $c<1$ depending on $\mu,\beta,\gamma$ and $d$ such that $1< \frac{\mu(\mathcal{T}_{t})}{\mu^*} \le \exp(c^t\iota).$
  \end{enumerate}
\end{proposition}

\begin{proof}
  (1) is obvious, we only prove (2).

  Note that $\mu(\mathcal{T}_t) = \mu h(\mu(\mathcal{T}_{t-1}))^d$, we denote $f(x)=\mu h(x)^d$. Recall that $\mu^*$ is the largest fixed point of $f(x)$ and $f(\mu)<\mu$, we have $0<f'(\mu^*)<1$. Define $g(x)=\frac{xf'(x)}{f(x)}$, then $g(\mu^*)=f'(\mu^*)$. Since $g(x)$ is a continuous function, we can choose some $\eta>0$ such that $0<g(x)\le c<1$ for all $x\in(\mu^*-\eta,\mu^*+\eta)$. 
  
  We now define a sequence $\set{x_i}_{i\ge 0}$ such that $x_0=\mu$ and $x_i=f(x_{i-1})$ for all $i\ge 1$. We claim that $\set{x_i}$ converges to $\mu^*$ as $i$ approaches infinity. To see this, note that $x_{i+1}=f(x_i)<x_i$ and $x_i>\mu^*$ for all $i\ge 0$. This implies $\{x_i\}$ converges to some $z\ge \mu^*$. Moreover, since $f$ is continuous, the sequence $\set{f(x_i)}_{i\ge 0}$ also converges to $z$. These two facts together imply $z = \lim_{i\to\infty}f(x_i) = f(\lim_{i\to\infty}x_i) = f(z)$. In other word, $z$ is a fixed point of $f$ and thus $z=\mu^*$. The claim implies that for some integer $t_0$, $x_{t_0}\in(\mu^*,\mu^*+\eta)$.
  
  We define another sequence $\set{y_i}_{i \ge 0}$ such that $y_0=\mu(\mathcal{T}_{t_0})$ and $y_i=f(y_{i-1})$ for all $i\ge 1$. It holds that $y_i\in (\mu^*,\mu^*+\eta)$ and thus $g(y_i)\le c<1$ for all $i\ge 0$. Therefore for all $t\ge 1$, 
  \begin{align*}
    \ln y_{t}-\ln\mu^*
    &=\ln f(y_{t-1})-\ln f(\mu^*)\\
    &=\frac{\tilde yf'(\tilde y)}{f(\tilde y)}\cdot\left|\ln y_{t-1}-\ln\mu^*\right|\quad\mbox{for some }y\in[\mu^*,y_{t-1}]\\
    &=g(\tilde y)\cdot \left|\ln y_{t-1}-\ln\mu^*\right|\\
    &\le c \cdot \left|\ln y_{t-1}-\ln\mu^*\right|\\
    &\le c^t\eta.
  \end{align*}
  We denote $\iota = \max\left\{\ln\mu, \eta c^{-t_0}\right\}$ and conclude the proof.
  \qed
\end{proof}

Our main idea to realize a target external field $\hat{\mu}$ is to construct a list of gadgets $\mathcal{G}=\{G_1,\dots,G_k\}$ such that $\mu\left(\mathtt{comb}(\mathcal{G})\right)\approx \hat{\mu}$ or more concretely $\hat{\mu} \approx \mu\prod_{i\in[k]}h\left(\mu(G_i)\right)$. All but one of these $G_i$ are basic gadgets of the following three types: (1) isolate point $\mathcal{S}_0$ with $\mu(\mathcal{S}_0)=\mu$; (2) $\mathcal{S}_w$ with large enough $w$ such that $\mu(\mathcal{S}_w)\approx 0$; and (3)$\mathcal{T}_t$ with large enough $t$ such that  $\mu(\mathcal{T}_t)\approx \mu^*$. The remaining one $G_i$ is recursively constructed with a new target  $\hat{\mu}'$ so that ideally  $\hat{\mu} =\mu\prod_{i\in[k]}h\left(\mu(G_i)\right)$ holds.
The combination of these basic gadgets are carefully chosen so that the new target $\hat{\mu}'$ is also in the range $(0, \mu^*]$. Then we recursively construct this $\hat{\mu}'$ by a subtree.
We terminate the recursion after enough steps, and use a basic star gadget which is closest to the desired value as an approximation in the leaf.
With a correlation decay argument, we show that the error in the root can be exponentially small in terms of the depth, although there may be a constant error in the leaf.
A detailed construction with special treatment for the boundary cases are formally given in  Algorithm \ref{algo:construct}.

\begin{algorithm}[ht]\label{algo:construct}
  \SetKwInOut{Input}{input}\SetKwInOut{Output}{output}
  \emph{ \textbf{function} $\mathtt{construct(\ell,\hat\mu)}:$}
  \BlankLine
  \Input{Recursion depth $\ell$; Target $0<\hat\mu\le \mu^*$ to simulate; }
  \Output{Graph $G_\ell$ constructed.}
  \Begin{
    \If{$\ell=0$ }{
      Let $k$ be the positive integer such that
      $\mu h(\mu)^{k+1}<\hat\mu\le\mu h(\mu)^{k}$\;
      \Return{ $S_k$}\;
    }
    \Else{
      Let $k$ be the non-negative integer with $\mu^*h(\mu)^{k+1}<\hat\mu\le\mu^*h(\mu)^k$ \;
      $\mathcal{Y}' \gets k \cdot \mathcal{S}_0$\tcp*[l]{a set of $k$ copies of $\mathcal{S}_0$.}
      $\mu_1 \gets \frac{\hat\mu}{h(\mu)^k}$ \;
      \tcp{Invariant: $\mu h(x')^{d-i+1} = \mu_i$ has a solution $0< x' \le \mu^*$.}
      \For{$i \gets 1$ \KwTo $d-1$}{
        \If{$\mu h(\mu^*) h(0)^{d-i} \ge \mu_i$}{
          $y_i \gets 0$;   $w \gets \lfloor \frac{\ell \cdot \ln \alpha - \ln (d \mu)}{\ln \beta} \rfloor + 1$;       $Y_i \gets \mathcal{S}_w$\;
        }
        \Else{
          $y_i \gets \mu^*$;
          $t \gets \lfloor\frac{\ell\cdot\ln\alpha-\ln d-\ln\iota}{\ln c}\rfloor+1$;
          $Y_i \gets \mathcal{T}_t$\;
        }
        $\mu_{i+1} \gets \frac{\mu_{i}}{h(y_i)}$\;
      }
      Let $\hat\mu'$ be the solution of $\mu h(x)=\mu_d$ in $(0,\mu^*]$\; 
      $\mathcal{Y} \gets \mathcal{Y}' \cup \set{Y_i}_{i\ge 1}^{d-1}$\;
      $\delta\gets\exp(-\frac{\ln\gamma\ln\alpha}{\ln\beta}\ell+\frac{\ln\gamma\ln(d\mu)}{\ln\beta}+\ln\frac{\mu}{\gamma})$\;
      \If{$\hat\mu'\le \delta$}{
        Choose the largest integer $w$ such that $\mu\left(\frac{1}{\gamma}\right)^w> \delta $\;
        \Return{$\mathtt{comb}(\mathcal{Y}\cup\{\mathcal{S}_w\})$}\;
      }
      \Else{

        \Return{ $\mathtt{comb}(\mathcal{Y}\cup\mathtt{construct(\ell-1,\hat\mu')})$}\;
      }
    }
  }
  \caption{Constructing $G_\ell$}
\end{algorithm}

\bigskip
Before we prove that the construction is correct, we obtain a few observations which is used in our proof.
The condition on $\mu$ in the key Lemma \ref{lemma:key-reduction} is due to the following property we need.

\begin{proposition}
  \label{prop:init-sol}
  Let $\mu > \frac{\gamma^d}{\beta} (\beta \gamma -1)\left(1+\frac{d+1}{\ln \left(\beta (\beta \gamma)^d\right)}\right)$, for any $\mu_1$ with $\mu^* h(\mu) < \mu_1 \le \mu^*$, the equation $\mu h(x)^d = \mu_1$ always has a solution with $0 < x \le\mu^*$.
\end{proposition}
\begin{proof}
  It suffices to show $\mu\cdot h(0)^d\le\mu^*h(\mu)$ and $\mu\cdot h(\mu^*)^d\ge \mu^*$.
  Since $\mu^* = \mu h(\mu^*)^d$, the second part is trivial.
  As for the first part,  it is sufficient to show
  $\left(\frac{h(\mu^*)}{h(0)}\right)^d h(\mu) > 1$.
  Note that $\left(\frac{h(\mu^*)}{h(0)}\right)^d h(\mu) >\gamma^d h(\mu^*)^{d+1} > \gamma^d \left(\beta - \frac{\beta \gamma - 1}{\mu^*}\right)^{d+1}$,
  \[
  \gamma^d \left(\beta - \frac{\beta \gamma - 1}{\mu^*}\right)^{d+1} > 1
  \iff   \ln \left(\beta (\beta \gamma )^d\right) + (d+1)\ln \left( 1 - \frac{\beta \gamma - 1}{\beta \mu^*} \right)  > 0,
  \]
  \[
  (d+1)\ln \left( 1 - \frac{\beta \gamma - 1}{\beta \mu^*} \right)  \overset{(\clubsuit)}{>} -(d+1) \frac{\frac{\beta \gamma - 1}{\beta \mu^*}}{1 - \frac{\beta \gamma - 1}{\beta \mu^*}} \overset{(\spadesuit)}{>} -\ln \left(\beta (\beta \gamma )^d\right) ,
  \]
  where $(\clubsuit)$ is due to $\ln(1-x) > -\frac{x}{1-x}$ for $x \in (0,1)$,
  and $(\spadesuit)$ is by the fact that $\beta (\beta \gamma)^d >1$ and the choice of $\mu$ such that $ - \frac{\beta \gamma - 1}{\beta \mu^*} >\frac{\ln \left(\beta (\beta \gamma)^d\right)+d+1}{\beta \ln \left(\beta (\beta \gamma)^d\right)}$.
\qed
\end{proof}

\begin{proposition}\label{prop:hproperty}
  For every $x,t\ge 0$, it holds that $h(x+t)\le (1+t)h(x)$ and $h\left((1+t)x\right)\le(1+t)h(x)$.
\end{proposition}
\begin{proof}
  Note that $x,t\ge 0$,
  \begin{align*}
    & h(x+t)\le (1+t)h(x)
    \iff \left(\frac{\beta(x+t)+1}{x+t+\gamma}\right)\le\left(1+t\right)\left(\frac{\beta x+1}{x+ \gamma}\right)\\
    \iff& t^2(1 + \beta x) + t\left(1 + \gamma \left(1 + \beta(x-1)\right) + x + \beta x^2\right)\ge 0.
  \end{align*}
  Since $\left(1 + \gamma \left(1 + \beta(x-1)\right) + x + \beta x^2\right)>0$, the inequality always holds.
  \begin{align*}
    h\left((1+t)x\right)\le (1+t)h(x)
    \iff& \frac{x(1+t)\beta+1}{x(1+t)+\gamma}\le (1+t)\frac{\beta x+1}{x+\gamma}\\
    \iff& t^2 (x + \beta x^2) + t (\gamma + 2x + \beta x^2)\ge 0
  \end{align*}
  Again every term is non-negative, the last inequality is always true.
  \qed
\end{proof}

In the following, we start to verify the correctness of the construction. We first verify that the algorithm is well defined, namely $\mu h(x)=\mu_d$ does have a solution  $\hat\mu'$  in $(0,\mu^*]$.
This can be done by verifying the loop invariant ``$\mu h(x')^{d-i+1} = \mu_i$ has a solution $0< x' \le \mu^*$" inductively.
\begin{description}
\item[Initialization.]
  For $i=1$, by Proposition \ref{prop:init-sol}, for some $0<\tilde x\le\mu^*$ it holds that $\mu h(\tilde x)^{d-i+1}=\mu_i$.

\item[Maintenance.]
Assuming $\mu h(\tilde x)^{d-i+1}=\mu_i$ has solutions $\tilde x\in(0,\mu^*]$,
  we verify that $\mu h(x')^{d-i}=\mu_{i+1} \equiv \frac{\mu_i}{h(y_i)}$ has solutions $x'\in(0,\mu^*]$ for $i \in[1,d-1]$.

  {\noindent\bfseries Case} $\mu h(\mu^*) h(0)^{d-i} \ge \mu_i$.
  By assumption we have $\mu h(0)^{d-i+1} < \mu_i$,
  also note that $\mu_i \le \mu h(\mu^*)h(0)^{d-i}\le \mu h(0)h(\mu^*)^{d-i}$,
  hence
  $\mu h(0)^{d-i} < \frac{\mu_i}{h(0)} \le \mu h(\mu^*)^{d-i}$.
  Then by continuity, $\mu h(x')^{d-i}=\frac{\mu_i}{h(0)}$  has solutions $0< x'\le\mu^*$.

  {\noindent\bfseries Case} $\mu h(\mu^*) h(0)^{d-i} < \mu_i$.
  By assumption $\mu h(\mu^*)^{d-i+1} \ge \mu_i$,
  thus $\mu h(0)^{d-i} < \frac{\mu_i}{h(\mu^*)} \le \mu h(\mu^*)^{d-i}$,
  hence $\mu h(x')^{d-i}=\frac{\mu_i}{h(\mu^*)}$  has solutions $0< x'\le\mu^*$.

\item[Termination.]
  After the loop completes, 
  $\mu h(x')=\mu_{d}$  has solutions $0< x'\le\mu^*$.
\end{description}


Now we verify the vertex weight gadget returned by the construction satisfies our requirement by choosing $\ell=O(-\log \eps) $.
\begin{lemma}
  \label{thm:gadget}
  For $0< \hat\mu \leq \mu^*(\beta,\gamma,\mu)$, and let $G(V,E)$ be the graph returned by $\mathtt{construct}(\ell,\hat\mu)$, we have the following: (1) $\exp(-(c+\ell)\cdot\alpha^{\ell})\le\frac{\mu(G)}{\hat\mu}\le\exp((c+\ell)\cdot\alpha^{\ell})$, where
  $c= \ln \gamma$ and $\alpha = \frac{\sqrt{\beta \gamma} - 1}{\sqrt{\beta \gamma} + 1}<1$; (2) $\abs{G}=\exp(O(\ell))$.
\end{lemma}
\begin{proof}
    We apply induction on $\ell$ for both statements. We prove for (1) first. For base case $\ell=0$, we have
    \[
    \left|\ln\mu(G)-\ln\hat\mu\right|
    \le \left|\ln\mu h(\mu)^k-\ln\mu h(\mu)^{k+1}\right|
    = -\ln h(\mu)
    \le\ln\gamma.
    \]
    Assume the statement holds for smaller $\ell$. Let $k$, $\set{y_i}_{1\le i\le d-1}$ and $\set{Y_i}_{1\le i\le d-1}$ be parameters chosen in the algorithm. Define
   \[      F(z) = \ln\left(\mu h(\mu)^k\prod_{i=1}^{d-1}h(y_i)h(\exp(z))\right),
      \tilde F(z)= \ln\left(\mu h(\mu)^k\prod_{i=1}^{d-1}h(\mu(Y_i))h(\exp(z))\right).\]
    We note that $F(z)$ is the \emph{correct} recursion to compute $\ln(\mu(G))$ and $\tilde F(z)$ is our \emph{approximate} recursion used in the algorithm.

    In the following, we distinguish between $\hat\mu'\le\delta$ and $\hat\mu'>\delta$.
    \begin{itemize}
    \item If $\hat\mu'\le\delta$, then $\ln\mu(G)=\tilde F(\ln\mu(\mathcal{S}_w))$ and $\ln\hat\mu=F(\ln\hat\mu')$. We have
      \begin{align*}
        F(\ln\hat\mu')\le\tilde F(\ln\mu(\mathcal{S}_w))
        &= \ln\left(\mu h(\mu)^k\prod_{i=1}^{d-1}h(\mu(Y_i))h(\mu(\mathcal{S}_w))\right)\\
        &\overset{(\heartsuit)}{\le} \alpha^\ell+\ln\left(\mu h(\mu)^k\prod_{i=1}^{d-1}h(y_i)h(\hat\mu')\right)\\
        &=\alpha^\ell+F(\ln\hat\mu')),
      \end{align*}
      where $(\heartsuit)$ follows from following facts derived from Proposition \ref{prop:hproperty}:
      \begin{enumerate}[(i)]
      \item If $y_i=0$, then $0\le\mu(Y_i)\le\frac{\alpha^\ell}{d}$, which implies $h(\mu(Y_i))\ge h(y_i)$ and $h(\mu(Y_i))\le \left(1+\frac{\alpha^\ell}{d}\right)h(y_i)\le\exp\left(\frac{\alpha^\ell}{d}\right)h(y_i)$.
      \item If $y_i=\mu^*$, then $\mu^*\le\mu(Y_i)\le\exp\left(\frac{\alpha^\ell}{d}\right)\mu^*$, which implies $h(\mu(Y_i))\ge h(y_i)$ and $h(\mu(Y_i))\le\exp\left(\frac{\alpha^\ell}{d}\right)h(y_i)$.
      \item We claim that $\hat\mu'< \mu\left(\frac{1}{\gamma}\right)^w\le\mu(\mathcal{S}_w)\le\mu\beta^w\le\frac{\alpha^\ell}{d}$. The only nontrivial part is to verify that $\mu\beta^w\le\frac{\alpha^\ell}{d}$. Since $w$ is the largest integer that $\hat\mu'<\mu\left(\frac{1}{\gamma}\right)^w$, we have $\mu\left(\frac{1}{\gamma}\right)^{w+1}\le\hat\mu'$, which gives $w\ge\frac{\ln\mu-\ln\delta}{\ln\gamma}-1$. Plug this into $\mu\beta^w\le\frac{\alpha^\ell}{d}$ and let $\delta=\exp(-\frac{\ln\gamma\ln\alpha}{\ln\beta}\ell+\frac{\ln\gamma\ln(d\mu)}{\ln\beta}+\ln\frac{\mu}{\gamma})$, the inquality holds. Thus $h(\mu(\mathcal{S}_w))\ge h(\hat\mu')$ and $h(\mu(\mathcal{S}_w))\le h(\frac{\alpha^\ell}{d})\le\left(1+\frac{\alpha^\ell}{d}\right)h(\hat\mu')\le\exp\left(\frac{\alpha^\ell}{d}\right)h(\hat\mu')$.
      \end{enumerate}
    \item If $\hat\mu'>\delta$, define $x=\mu(\mathtt{construct(\ell-1,\hat\mu')})$, then by induction hypothesis, it holds that $\left|\ln x-\ln\hat\mu'\right|\le (c+(\ell-1))\cdot\alpha^{\ell-1}$.

Then similarly by Proposition \ref{prop:hproperty} and the choice of $w$ and $t$, we have
$F(\ln x)\le \tilde F(\ln x)\le F(\ln x) + \alpha^\ell$.
Thus by construction, we have
\begin{align*}
  \left|\ln\mu(G)-\ln\hat\mu\right|
  &=\left|\tilde F(\ln x)-F(\ln\hat\mu')\right|\\
  &\le \alpha^\ell + \left|F(\ln x)-F(\ln\hat\mu')\right|\\
  &\le \alpha^\ell + \left|F'(\ln\tilde x)\right|\cdot\left|\ln x-\ln\hat\mu'\right| \ \ \ \hbox{(for some $\tilde x\in[\hat\mu',x]$.)}\\
  &\le \alpha^\ell+(\ell-1)\left|F'(\ln\tilde x)\right|\alpha^{\ell-1}+ c\left|F'(\ln\tilde x)\right|\alpha^{\ell-1}
\end{align*}
 Thus it is sufficient to show that
      $
      \left|F'(\ln\tilde x)\right|\le\alpha.
      $
      In fact,
      $F'(\ln x) = \frac{x\cdot h'(x)}{h(x)} = \frac{(\beta\gamma-1)x}{(x+\gamma)(\beta x+1)}\le\frac{\beta\gamma-1}{\left(\sqrt{\beta\gamma}+1\right)^2} = \alpha$.
    \end{itemize}

  Now we prove (2) of the Lemma. We denote $s(\ell)=\max_{\hat\mu}\abs{\mathtt{construct}(\ell,\hat\mu)}$ and show that $s(\ell) = \ell\exp(O(\ell)) =\exp(O(\ell))$.

    If $\ell=0$, since $\hat\mu$ is either the eventual external field (which is a constant bounded away from $0$), or $\hat\mu>\delta$, we have $s(\ell)=\abs{\mathcal{S}_k}=O(1)$.

    If $\ell>0$, then $\abs{Y_i}=\exp(O(\ell))$ and thus $\abs{\mathcal{Y}}=\exp(O(\ell))$. By our choice of $\delta$, it holds that $w=O(\ell)$ and thus $\abs{S_w}=O(\ell)$. Therefore,
    \[
    s(\ell) = \exp(O(\ell))+\max\set{s(\ell-1),O(\ell)} = \ell\exp(O(\ell))=\exp(O(\ell)).
    \]
    This concludes the proof.
    \qed
\end{proof}

\bibliographystyle{plain}
\bibliography{ferrohard}

\appendix

\section{The $\beta>1$ case}
We prove the hardness result for the case of $\beta>1$. This follows a similar argument as that in~\cite{goldberg2007complexity} and is known as a folklore. We include a formal proof here to be self-contain.  

\begin{theorem}
  \label{thm:large-beta}
  Let $\gamma>\beta> 1$ and $\mu>\frac{\gamma-1}{\beta-1}$. Then $\ferro(\beta, \gamma, \mu)$ is $\BIS$-hard. 
\end{theorem}





We follow the same idea of simulating external field and making use of Theorem \ref{thm:existence} to conclude the proof. In this case, we can simulate all positive external fields.

\begin{lemma}\label{lem:realizevertex}
  For every $\hat\mu>0$, there is a family of vertex weight gadgets $\set{G_m }_{m\ge 1}$ that realizes $\hat\mu$. Moreover, $G_m$ is constructible in time $m^{O(1)}$ and
  \begin{equation}
    \label{eq:vertexweight}
    \exp(-\frac{1}{m})\le\frac{\mu(G_m)}{\hat\mu}\le\exp(\frac{1}{m}).
  \end{equation}
\end{lemma}
\begin{proof}
  For any $m\ge 1$, we add $x$ self-loops and $y$ bristles to a single vertex $v$, where $x$ and $y$ are integers to be determined. Let $v$ be the output of  $G_m$, then $\mu(G_m)=\mu\left(\frac{\beta}{\gamma}\right)^x\left(\frac{\mu\beta+1}{\mu+\gamma}\right)^y$. Denote $a=\ln\frac{\gamma}{\beta}$, $b=\ln\frac{\mu\beta+1}{\mu+\gamma}$ and $c=\frac{\ln\hat\mu}{\ln\mu}$, then \eqref{eq:vertexweight} is equivalent to
  \[
  \left|\left(y\cdot b-x\cdot a\right)-c\right|\le\frac{1}{m}.
  \]
  We can use a procedure similar to extended Euclidean algoroithm to find such integers $x,y$ in time $O(\ln m)$, such that it also guarantees $x,y=m^{O(1)}$.
\end{proof}



%
%

\section{Improved Tractable Result}

In this section, we establish the following tractable result:

\begin{theorem}
Let $\beta<\gamma$, $\beta \gamma>1$ and $\mu\le\gamma/\beta$. 
Then there is an FPRAS for $\ferro(\beta,\gamma,\mu)$.
\end{theorem}

The proof of this theorem follows by refining the proof in \cite{goldberg2003computational}, where they establish the tractable result for $\mu\le\left(\gamma/\beta\right)^{\delta/2}$ for $\delta$ being the \emph{minimum} degree of vertices in the graph. Specifically, we first contract all vertices with degree one and modify the external fields of their neighboring vertices,  this only scales the partition function by a constant.
Next, just as in~\cite{goldberg2003computational}, we shall reduce a $(\beta,\gamma,\mu)$ instance to a ferromagnetic Ising instance and apply the following celebrated result, which is first introduced in~\cite{JS93} for uniform external fields and refined for non-uniform external fields in~\cite{goldberg2003computational}:

\begin{theorem}[\cite{JS93} and \cite{goldberg2003computational}]
  \label{thm:isingfpras}
  There is an FPRAS for Ising system $(a,a,\mathcal{V})$ provided that $a>1$ and all external fields in $\mathcal{V}$ are at most one.
\end{theorem}

Let $G(V,E)$ be an instance of $(\beta,\gamma,\mu)$ system, we repeatedly apply the following operations until no degree one vertices can be found:
\begin{enumerate}
\item Pick a vertex $u$ of degree one. Denote its incident edge by $e=(u,v)$. Let $\mu_u$ and $\mu_v$ be external fields on $u$ and $v$ respectively.
\item Remove $u$ and edge $(u,v)$, update $\mu_v \gets \mu_vh(\mu_u)$.
\end{enumerate}

Let $G'(V',E')$ be the remaining graph. $G'$ either has no vertices of degree one, or it only contains a single vertex.
Moreover, for every $v\in V'$, the external fields $\mu'_v$ satisfies $\mu'_v\le\mu$. This can be easily verified given that $\mu\le\gamma/\beta$. 
Let $\mathcal{U}=\{\mu'_v\mid v\in V'\}$, consider $G'$ as an instance of $(\beta,\gamma,\mathcal{U})$ system, clearly $Z_{(\beta,\gamma,\mu)}(G)=Z^*\cdot Z_{(\beta,\gamma,\mathcal{U})}(G')$ where $Z^*$ is an easily polynomial-time computable factor.

Let $\mathcal{V}=\set{\mu'_v\left(\frac{\beta}{\gamma}\right)^{d_v/2}\Big| v\in V'}$ where $d_v$ is the degree of $v$ in $G'$. Let $\hat G(\hat V,\hat E)$ be a copy of $G'$ with $\hat\mu_v=\mu'_v\left(\frac{\beta}{\gamma}\right)^{d_v/2}$ for every $v\in\hat V$. We are going to verify that $Z_{(\beta,\gamma,\mathcal{U})}(G')=\sqrt{\frac{\gamma}{\beta}}^{|E'|}\cdot Z_{(a,a,\mathcal{V})}(\hat G)$ for $a=\sqrt{\beta\gamma}$.

Define $A=
\begin{bmatrix}
  \beta & 1\\
  1 & \gamma
\end{bmatrix}
$,
$A'=
\begin{bmatrix}
  \gamma & \sqrt{\gamma/\beta}\\
  \sqrt{\gamma/\beta} & \gamma
\end{bmatrix}
$ and
$ \hat A=
\begin{bmatrix}
  \sqrt{\beta\gamma} & 1\\
  1 & \sqrt{\beta\gamma}
\end{bmatrix}
$.
Then,

\begin{align*}
  Z_{(\beta,\gamma,\mathcal{U})}(G')
  &=\sum_{\sigma\in\{0,1\}^{V'}}\prod_{(u,v)\in E'}A_{\sigma_u,\sigma_v}\prod_{v\in V'}{\mu'}_v^{1-\sigma_v}\\
  &=\sum_{\sigma\in\{0,1\}^{V'}}\prod_{(u,v)\in E'}A'_{\sigma_u,\sigma_v}\prod_{v\in V'}\left(\left(\sqrt{\frac{\beta}{\gamma}}\right)^{d_v}{\mu'}_v\right)^{1-\sigma_v}\\
  &=\sqrt{\frac{\gamma}{\beta}}^{|E'|}\sum_{\sigma\in\{0,1\}^{V'}}\prod_{(u,v)\in E'}\hat A_{\sigma_u,\sigma_v}\prod_{v\in V'}\left(\left(\sqrt{\frac{\beta}{\gamma}}\right)^{d_v}{\mu'}_v\right)^{1-\sigma_v}\\
  &=\sqrt{\frac{\gamma}{\beta}}^{|E'|}\sum_{\sigma\in\{0,1\}^{\hat V}}\prod_{(u,v)\in \hat E}\hat A_{\sigma_u,\sigma_v}\prod_{v\in V'}\hat\mu_v^{1-\sigma_v}\\
  &=\sqrt{\frac{\gamma}{\beta}}^{|E'|}\cdot Z_{(a,a,\mathcal{V})}(\hat G)
\end{align*}

Finally, to apply Theorem \ref{thm:isingfpras}, we only need $\hat\mu_v \le 1$ for all $v\in \hat V$.
Recall that $\hat G$ has $\delta \ge 2$, hence $\mu\le\gamma/\beta$ implies $\hat\mu_v \le 1$.
To sum up, this concludes the proof.

\end{document}